\documentclass{article}

\usepackage{PRIMEarxiv}

\usepackage[utf8]{inputenc} 
\usepackage[T1]{fontenc}    
\usepackage{hyperref}       
\usepackage{url}            
\usepackage{booktabs}       
\usepackage{amsfonts}       
\usepackage{nicefrac}       
\usepackage{microtype}      
\usepackage{lipsum}
\usepackage{fancyhdr}       
\usepackage{graphicx}       
\graphicspath{{media/}}     

\usepackage{amsthm}
\usepackage[english]{babel}
\usepackage{amsmath}
\usepackage{amssymb}
\usepackage{caption}
\usepackage{subcaption}
\usepackage[natbibapa]{apacite}
\usepackage{float}
\usepackage{tikz}
\usepackage{algorithm}
\usepackage{algpseudocode}
\usetikzlibrary{calc, matrix, arrows, shapes, positioning, fit, shapes.misc, shapes.geometric, calc, decorations.pathreplacing}
\newcommand{\norm}[1]{\Vert{#1}\Vert}
\usepackage{multirow}
\usepackage{dblfloatfix}

\newtheorem{theorem}{Theorem}[section]
\newtheorem{rem}{Remark}

\newtheorem{lemma}[theorem]{Lemma}
\newcommand{\M}[1]{\mathbf{#1}}
\newtheorem{algo}{Algorithm}
    \newtheorem{ass}{Assumption}
    
\newtheorem{coro}{Corollary}

\title{Fitting Low-rank Models on Egocentrically Sampled Partial Networks
}

\author{
  Ga Ming Angus Chan \\
  Department of Statistics \\
  University of Virginia \\
  \texttt{gc8ev@virginia.edu} \\
   \And
  Tianxi Li \\
  Department of Statistics \\
  University of Virginia \\
  \texttt{tianxili@virginia.edu} \\
}

\begin{document}
\newcommand\rank{\operatorname{rank}}
\newcommand\diag{\operatorname{diag}}
\maketitle

\begin{abstract}
  The statistical modeling of random networks has been widely used to uncover interaction mechanisms in complex systems and to predict unobserved links in real-world networks. In many applications, network connections are collected via egocentric sampling: a subset of nodes is sampled first, after which all links involving this subset are recorded; all other information is missing.  Compared with the assumption of ``uniformly missing at random", egocentrically sampled partial networks require specially designed modeling strategies. Current statistical methods are either computationally infeasible or based on intuitive designs without theoretical justification. Here, we propose an approach to fit general low-rank models for egocentrically sampled networks, which include several popular network models. This method is based on graph spectral properties and is computationally efficient for large-scale networks. It results in consistent recovery of missing subnetworks due to egocentric sampling for sparse networks. To our knowledge, this method offers the first theoretical guarantee for egocentric partial network estimation in the scope of low-rank models. We evaluate the technique on several synthetic and real-world networks and show that it delivers competitive performance in link prediction tasks. 

\end{abstract}

\section{INTRODUCTION}\label{sec:intro}

Massive network data that capture complicated dynamics and interactions in human society, the economy, ecosystems, and biology are now available \citep{goldenberg2010survey,newman2018networks}. The past 15 years have witnessed substantial progress in random network models within the statistics field. Associated efforts have provided countless model options to analyze network data with well-established theories \citep{hoff2002latent,bickel2009nonparametric,rohe2011spectral,gao2021minimax}. Drawing insights from complex networks is fundamental to many scientific challenges  \citep{kolaczyk2014statistical,newman2018networks}, such as understanding community structures \citep{karrer2011stochastic}, predicting new links \citep{liben2007link,zhao2017link}, and modeling peer effects in downstream tasks \citep{le2020linear}.

Missing data is a commonly encountered issue in data analysis \citep{little2019statistical} and plagues network problems, especially in social networks. Network connections are frequently obtained through surveys or sampling processes. One application of the network model, link prediction, is intrinsically embedded in missing data scenarios \citep{martinez2016survey,kumar2020link}. Moreover, the missingness in network problems often exhibits unique patterns and calls for specialized modeling strategies. 

We consider the missingness from \emph{egocentric sampling} in this paper. Egocentric sampling is a widely used mechanism for acquiring network data  \citep{bandiera2006social,ali2009estimating, alatas2016network, arnaboldi2013egocentric}. Under this approach, a subset of subjects is randomly sampled and all their connections are recorded. Any connections between subjects outside the sample are missing. \cite{handcock2010modeling} designs a model based on the exponential random graph model (ERGM), which can handle egocentrically missing data. However, their model fitting is not computationally feasible for moderately sized networks in general, though its computation can be improved in certain settings  \citep{krivitsky2017inference}. As we will show, this class of models is too restrictive to make effective link predictions. \cite{wu2018link} introduced an algorithm motivated by the CUR decomposition in computational mathematics \citep{Drineas2006a}, which is computationally efficient and demonstrates strong empirical performance on link prediction tasks. However, the underlying statistical model fitted by their method is unclear. Theoretical guarantees are not available for the aforementioned methods. The only family with known theoretical model-fitting correctness is the stochastic block model family, implicitly available from \cite{chen2018network}, including \cite{chandrasekhar2011econometrics} as a special case.

This paper proposes a method to estimate general low-rank random network models based on egocentrically sampled partial networks. Our technique can consistently estimate network models, even for sparse networks whose average node degree is sublinear in sample size. To our knowledge, this approach is the first to feature theoretical guarantees for general low-rank models on egocentrically sampled partial networks. Our results cover many special cases, most of which previously had no known model-fitting theory; examples include the setting of \citet{wu2018link}, along with the random dot product model  \citep{young2007random} and its generalization \citep{rubin2017statistical}. As an unexpected byproduct, our method provides a new insight into the method of \cite{wu2018link}: while \cite{wu2018link} motivated their algorithm by a CUR format, their method indeed fits a general low-rank structure. Additionally, our algorithm is based on the spectral decomposition of the partial network adjacency matrix, a technique that is extremely efficient for computation on large-scale networks. We empirically demonstrate that our approach displays competitive performance in dealing with link prediction problems.

\section{METHODOLOGY}

\subsection{Setup}

\paragraph{Notations.} We will use bold font capital letters, such as $\M{A}$, to denote a matrix. $\M{A}_{kk'}$ will be used to denote a submatrix of $\M{A}$ (to be defined later), while the element at the $i$th row and $j$th column of $\M{A}$ will be denoted by $A_{ij}$. Let $\M{A}^T$ and $\M{A}^+$ be the transpose and Moore-Penrose inverse of $\M{A}$, respectively. Furthermore, $\norm{\M{A}}_F$ is the Frobenius norm of $\M{A}$. For any positive integer $n$, we define $[n] = \{1, 2, \cdots, n\}$.

Let $N$ be the total number of nodes in the full network, indexed by $i=1, \cdots, n$. We can represent the network by its adjacency matrix $\M{A}\in\{0,1\}^{N\times N}$, where $A_{ij} = 1$ if and only if nodes $i$ and $j$ are connected in the network. We consider undirected and unweighted networks for presentation simplicity. In this case, we have $\M{A}^T = \M{A}$. In Section~\ref{secsec:extension}, we briefly discuss how to extend our method to handle more general networks. We will study the statistical properties under the so-called ``inhomogeneous Erd\"{o}s-Renyi framework". Specifically, we assume there exists a probability matrix $P\in [0,1]^{N\times N}$ such that
$$A_{ij} \sim \mathrm{Bernoulli}(P_{ij}),~~ i < j ~~~~\text{independently}.$$
This framework is arguably one of the most prominently applied for random network modeling. Under it, the assumed network structures are incorporated into matrix  $P$. We assume a low-rank model for our study and define  $K = \rank(P) \ll N$. The family of low-rank models includes many prevalent random network models such as the stochastic block model (SBM), the degree-corrected block model (DCBM) and their generalizations  \citep{holland1983stochastic,airoldi2008mixed,karrer2011stochastic, jin2017estimating,sengupta2018block,li2022hierarchical} as well as the random dot product model (RDPG) and its variants  \citep{young2007random, rubin2017statistical}. Indeed, as studied by \cite{chatterjee2015matrix}, most popular random network models are approximately low-rank.

When the network is only partially observed from egocentric sampling, suppose there are  $n$ nodes whose connections are fully observed. Without loss of generality, we can assume that these nodes are the first $n$ rows and columns in the adjacency matrix $\M{A}$. Consider the following block partition:
\begin{equation*}\label{matBlock}
    \M{A} = 
    \left(\begin{array}{cc}
    \mathbf{A}_{11}&\mathbf{A}_{12}\\
    \mathbf{A}_{21}&\textcolor{red}{\mathbf{A}_{22}}\\
    \end{array}\right),
\end{equation*}
where $\mathbf{A}_{11}\in\{0,1\}^{n\times n}$, $\mathbf{A}_{12}\in\{0,1\}^{n\times (N-n)}$, $\mathbf{A}_{21}=\mathbf{A}_{12}^\intercal$, and $\mathbf{A}_{22}\in\{0,1\}^{(N-n)\times(N-n)}$. Note that $\M{A}_{12} = \M{A}_{21}^T$. In our problem, $\mathbf{A}_{11}$, $\mathbf{A}_{12}$ and $\mathbf{A}_{21}$ are observed while $\mathbf{A}_{22}$ (in red) is missing.  We can partition $\M{P}$ in the same way
\begin{equation*}\label{matBlock_2}
    \M{P} = 
    \left(\begin{array}{cc}
    \mathbf{P}_{11}&\mathbf{P}_{12}\\
    \mathbf{P}_{21}&\mathbf{P}_{22}\\
    \end{array}\right).
\end{equation*}
The goal of model fitting is to recover $\M{P}$ from the observed blocks $\mathbf{A}_{11}$, $\mathbf{A}_{12}$ and $\mathbf{A}_{21}$. The core challenge is to recover $\M{P}_{22}$ for which no observations are available.

\subsection{Low-rank Estimation (LE) Algorithm}

A natural approach to the current problem is to use certain types of low-rank matrix completion \citep{Candes2010a,Plan2011}. However, as shown in \cite{wu2018link}, such methods can be slow and suffer from poor accuracy due to the egocentric missing pattern. \cite{wu2018link} employed an intuitive CUR decomposition based on the missing structure of  $\M{A}$ that efficiently computes a low-rank imputation with good empirical performance.  Yet, the exact reason for this technique’s success is unclear. We take a more principled approach by leveraging the low-rank structure precisely, leading to superior performance over \cite{wu2018link} while also providing theoretical guarantees.

Specifically, our algorithm is motivated by a self-consistency property studied by \citet{owen2009bi} for low-rank matrices.
\begin{lemma}[\citet{owen2009bi}]\label{lemma:owen}
For any $p\times q$ matrix $\M{M}$ with the partition
\begin{equation*}
    \M{M} = 
    \left(\begin{array}{cc}
    \mathbf{M}_{11}&\mathbf{M}_{12}\\
    \mathbf{M}_{21}&\mathbf{M}_{22}\\
    \end{array}\right),
\end{equation*}
Suppose $\rank(\M{M}_{11}) = \rank(\M{M})$, we have
    \begin{equation*}
        \mathbf{M}_{11}
        = \mathbf{M}_{12}\mathbf{M}_{22}^+\mathbf{M}_{21}.
    \end{equation*}
\end{lemma}

As per this lemma, suppose $\rank(\M{P}_{11}) = K$. We can then exactly compute $\M{P}_{22}$ by $\M{P}_{11}$ and $\M{P}_{12}$. In practice, when we do not observe $\M{P}_{11}$ and $\M{P}_{12}$, it seems logical to take $\M{A}_{11}$ and $\M{A}_{12}$ as the ``plug-in" estimators of $\M{P}_{11}$ and $\M{P}_{12}$. Yet this naive approach will not work well for two reasons.  First, due to the binary nature, the Moore-Penrose inverse operator on $\M{A}_{11}$ is too noisy to be a good approximation. Second, directly using $\M{A}_{11}$ ignores the requirement of $\rank(\M{P}_{11}) = K$ needed in Lemma~\ref{lemma:owen}. An additional step is thus necessary to resolve the two issues simultaneously. In brief, we first take the optimal rank-$K$ approximation of $\M{A}_{11}$ as a smoothing step and then take the Moore-Penrose of the smoothed estimator. The full procedure is described in  Algorithm~\ref{algo:LE}.
\begin{algo}[LE imputation for the missing network]\label{algo:LE}
Given egocentrically sampled submatrices $\M{A}_{11}$, $\M{A}_{12}$, and the rank $K$, complete the following steps:
\begin{enumerate}
    \item Take the singular value decomposition of $\mathbf{A}_{11}=\mathbf{UDV}^T$, where $\M{D}$ is the diagonal matrix containing the singular values of $\M{A}_{11}$ in non-increasing order, and $\M{U}$ and $\M{V}$ are orthogonal matrices with each column being a singular vector.
    \item Compute $\tilde{\mathbf{P}}_{11}=\mathbf{U}_{(K)}\M{D}_{(K)}\M{V}_{(K)}^T$, in which $\mathbf{U}_{(K)}$ and $\mathbf{V}_{(K)}$ are the matrices of the first $K$ columns of $\M{U}$ and $\M{V}$, respectively, and $\M{D}_{(K)}$ is the diagonal matrix of only the first $K$ singular values. 
    \item Set $\hat{\mathbf{P}}_{22} = \mathbf{A}_{12}^T\tilde{\mathbf{P}}_{11}^+\mathbf{A}_{12}$. 
    \item (Optional) If a strict constraint of all entries within $[0,1]$ is desired, truncate all values of $\hat{\mathbf{P}}_{22}$ to $[0,1]$.
    \item Return $\hat{\mathbf{P}}_{22}$.
\end{enumerate}
\end{algo}

\subsection{Connections and New Insights to \cite{wu2018link}}

Algorithm~\ref{algo:LE} turns out to be closely connected to the subspace estimation (SE) method of \cite{wu2018link}. In particular, suppose we replace the SVD of Steps 1-2 in Algorithm~\ref{algo:LE} by the rank-$K$ approximation of $\M{A}_{\mathrm{obs}} = (\M{A}_{11}, \M{A}_{12})$ and then take the resulting first $n$ columns as $\hat{\M{P}}_{11}$. Using $\hat{\M{P}}_{11}$ instead of  $\tilde{\mathbf{P}}_{11}$ in other calculations would lead to the SE method. This connection is a bit surprising given that the SE method is originally designed for matrices with CUR decomposition (see their Theorem 1), a strict subset of low-rank matrices. Our current connection, therefore, reveals a new interpretation for the SE method -- it is fitting general low-rank models instead of CUR-form models. As we will show, the LE method comes with theoretical guarantees for its performance. However, our analysis could not be extended to the SE method. This is because the column extraction operation of $\hat{\M{P}}_{11}$ complicates the perturbation analysis. Our intuition is that such a step, though it uses more information from data (by taking $\M{A}_{\mathrm{obs}}$), also requires a strong signal-to-noise ratio. Our empirical experiments support this conjecture (see Section~\ref{sec:simulation}). We leave the theoretical analysis of the SE estimation for future work.

\subsection{Other Considerations}\label{secsec:extension}

\textbf{Model tuning by cross-validation.} 
The LE algorithm takes the rank $K$ as a tuning parameter. In practice, we can tune a proper $K$ via cross-validation as indicated in \cite{wu2018link}. Specifically, we randomly hold out $\rho$ proportion of the fully observed nodes, denoted by $V \subset [n]$ for validation. The remaining nodes $[n]\setminus V$ are treated as a smaller egocentrically sampled partial network, on which the LE algorithm is applied with a sequence of $K$. We then evaluate the link prediction performance of the LE algorithm with each $K$ value on the partial network associated with the hold-out set $V$. The $K$ value that achieves the highest area under the receiver operating characteristic (ROC) curve (i.e., AUC) is returned. This procedure is repeated multiple times, and we use the rounded average of the selected ranks as the final rank.

\textbf{Full matrix recovery.} 
Algorithm~\ref{algo:LE} only outputs the estimated $\hat{\M{P}}_{22}$. This is because imputing the missing subnetwork would be the most widely needed step in analyzing partial network data, and it would be the unique contribution of our paper. Estimation of the full matrix $\M{P}$, as a global model estimation, is also available. The other components $\M{P}_{11}$ and $\M{P}_{12}$ can be estimated together with the low-rank smoothing of the observed component $\M{A}_{\mathrm{obs}}$, by taking its rank-$K$ truncated SVD. Suppose the resulting matrix is $\tilde{\M{P}}_{\mathrm{obs}}$. We can take its last $N-n$ columns $\tilde{\M{P}}_{\mathrm{obs},(n+1):N}$ as $\hat{\M{P}}_{12}$ and the first $n$ column as  $\hat{\M{P}}_{11}$. Then due to symmetric properties, we set $\hat{\M{P}}_{21} = \hat{\M{P}}_{12}^T$. We can also force symmetry on $\hat{\M{P}}_{11}$ by taking $\frac{1}{2}(\hat{\M{P}}_{11}+\hat{\M{P}}_{11})^T$ as the estimator of $\M{P}_{11}$. Similar to Step 4 of Algorithm~\ref{algo:LE}, we can truncate the values of all estimators to $[0,1]$ as an optional step. Our theoretical properties hold in both cases. Combined with the output $\hat{\M{P}}_{22}$ of Algorithm~\ref{algo:LE}, the final estimation of the full matrix $\M{P}$ is given by
\begin{equation}\label{eq:full-matrix}
    \hat{\M{P}} = 
    \left(\begin{array}{cc}
    \hat{\mathbf{P}}_{11}&\hat{\mathbf{P}}_{12}\\
    \hat{\mathbf{P}}_{21}&\hat{\mathbf{P}}_{22}\\
    \end{array}\right).
\end{equation}
As we will show in the next section, $\hat{\M{P}}$ would be a consistent estimator of $\M{P}$. However, in practice, $\M{P}$ may not be an exactly rank-$K$ matrix. We do not think it is necessary to force the rank of $\hat{\M{P}}$ in most applications. However, if needed, one can force the rank value with the post-processing of SVD truncation.

\textbf{Extension to directed and weighted networks.}
While we focus on undirected and unweighted networks, the extension to directed and weighted networks can be naturally made. Notice that Lemma~\ref{lemma:owen} holds for general asymmetric matrices. $\M{A}$ and $\M P$ are no longer symmetric for a directed network. And the inhomogeneous Erd\"{o}s-Renyi model is still well defined, with all entries $A_{ij}$ being independently generated according to $P_{ij}$. Algorithm~\ref{algo:LE} can be exactly applied. Our theoretical results can be similarly derived. 

\section{THEORETICAL PROPERTIES OF THE LE ESTIMATOR}\label{sec:theory}

As mentioned, one major contribution of our work is the theoretical guarantee for general low-rank network model recovery from the partial network, which is not previously available. We now proceed to introduce our theoretical results. We first introduce additional notations and our regularity assumptions. Define $p^* = \max_{ij}P_{ij}$. For any matrix $\M{M}$, denote its $k$th largest singular value by $\sigma_k(\M{M})$. For two sequences $a_n$ and $b_n$, we will write $a_n \lesssim b_n$ if there exists a constant $C$ such that $a_n \le C b_n$. Correspondingly, we write $b_n \gtrsim a_n$ if $a_n \lesssim b_n$. Moreover, if $a_n \lesssim b_n$ and $b_n \lesssim a_n$, we write $a_n \sim b_n$.

We make the following assumptions. 
\begin{ass}[Low-rank recoverable]\label{ass:rank}
The rank of the model satisfies $\rank(\M{P}_{11}) = \rank(\M{P}) = K$. In particular, $K \lesssim \sqrt{\log n}$.
\end{ass}
\begin{ass}[Well-conditioned model]\label{ass:condition}
    $$np^* \sim \sigma_{K}(\M{P}_{11}) \le \sigma_{1}(\M{P}_{11}) \sim np^*$$
    $$Np^* \sim \sigma_{K}(\M{P}) \le \sigma_{1}(\M{P}) \sim Np^*$$
\end{ass}
Assumption~\ref{ass:rank} is strictly needed to ensure the validity of the low-rank recovery on the population matrix $\M{P}$ by Lemma~\ref{lemma:owen}. In contrast, assumption \ref{ass:condition} can be relaxed for better generality. However, we keep it in the current form for conciseness and interpretability of our error bound. 
\begin{rem}Both \ref{ass:rank} and \ref{ass:condition} are indeed motivated by the general sparse graphon model of \cite{bickel2009nonparametric}. The sparse graphon framework is arguably the most popular way to generate the $\M{P}$ in the inhomogeneous Erd\"{o}s-Renyi model \citep{klopp2017oracle,mukherjee2017clustering,lin2020theoretical,gao2021minimax}. In particular, it is not difficult to see that under the sparse graphon mechanism, if the true graphon does generate a low-rank $\M{P}$ and the observed nodes are uniformly sampled in the egocentric sampling, both \ref{ass:rank} and \ref{ass:condition} hold almost surely. In particular, the assumptions hold for the SBM, RDPG and their variants.
\end{rem}
\begin{rem}
Meanwhile, we also want to emphasize that we do not assume that observed subjects form a random sample out of all nodes. So our theory is still applicable in situations of nonuniform sampling or even informative sampling, as long as \ref{ass:condition} holds. This property is crucial for real-world applications. We will provide evidence supporting this claim in Section~\ref{sec:simulation}.
\end{rem}

\begin{theorem}\label{thm:frobenius}
Under assumptions~\ref{ass:rank} and \ref{ass:condition}, further assume that $np^*\gtrsim  \mathrm{log}(n)$ and $K \lesssim \sqrt{\log n}$ . If $\hat{\M P}_{22}$ is the estimator of $\M{P}_{22}$ produced by Algorithm~\ref{algo:LE}, we have  
\begin{equation}
    \norm{\hat{\M P}_{22}-\M P_{22}}_F \lesssim \sqrt{K}\left(\left(\frac{N}{n}\right)^{3/2}\sqrt{KNp^*} + \frac{N^2}{n^2}\log n\right)
\end{equation}
with probability tending to 1.
\end{theorem}

The requirement for $K$ can be further relaxed with a more complicated error bound which we will not pursue. For illustration, consider the following two special cases
\begin{enumerate}
\item Suppose $n$ and $N$ are in the same order ($n\sim N$), we can see that the error bound on the missing network is in the order of $K\sqrt{np^*}+\sqrt{K}\log n$. Since $\norm{\M{P}_{22}}_F \sim n\sqrt{p^*}$, we know that $\norm{\hat{\M P}_{22}-\M P_{22}}_F/\norm{\M{P}_{22}}_F \to 0$ and the estimation consistency is guaranteed under the current assumptions.
\item Suppose $K$ is bounded and $np^* = \log^2 n$. Then the error bound is the order of $\frac{N^2}{n^2}\log n$. So $\norm{\hat{\M P}_{22}-\M P_{22}}_F/\norm{\M{P}_{22}}_F \to 0$ as long as $n \gg N^{4/5}$. Therefore, though our method allows the sampling proportion to be vanishing, the decaying rate has to be slow.
\end{enumerate}

The error bound for the full matrix recovery is a straightforward extension of Theorem~\ref{thm:frobenius} as follows.
\begin{coro}\label{coro:full}
Under the assumptions of Theorem~\ref{thm:frobenius}, for the full matrix estimator $\hat{\M{P}}$ defined in \eqref{eq:full-matrix}, we have  
$$\norm{\hat{\M P}-\M P}_F \lesssim \sqrt{K}\left(\left(\frac{N}{n}\right)^{3/2}\sqrt{KNp^*} + \frac{N^2}{n^2}\log n\right)$$
with high probability.
\end{coro}
Corollary~\ref{coro:full} still has the same order error bound as  Theorem~\ref{thm:frobenius}. This indicates that the major estimation error for the full matrix $\M P$ is still on the unobserved component $\M P_{22}$. Because the components $\M P_{11}$ and $\M P_{12}$ are observed (with noises) and the model is low-rank. The recovery of these terms would be easy. The imputation of the unobserved subnetwork is the most challenging component of the problem, and it is our emphasized contribution. Similar to the previous illustration, the corollary indicates that if $N\sim n$, for example, the full matrix estimation is consistent.

\section{SYNTHETIC EXPERIMENTS}\label{sec:simulation}

In this section, we evaluate the proposed LE method in link prediction tasks to predict the unobserved subnetwork under several synthetic network models. Our implementation is based on R. We include all the egocentric-suitable methods mentioned in Section~\ref{sec:intro} as benchmarks as follows.

\begin{enumerate}
    \item Subspace estimation method (SE) \citep{wu2018link}. This method is based on the CUR decomposition of a matrix, as a special case of the low-rank model. The approach is effective for link prediction with computational scalability. It does not come with theory.
    \item LE+. This method is computed as the average of the LE and SE estimations. This simple hybrid will be adaptive to better one of the two methods according to individual scenarios.
    \item The neighborhood smoothing method (NS) for graphon estimation \citep{zhang2017estimating}. This method was originally studied under piecewise-smooth dense graphon models  \citep{bickel2009nonparametric}, and we adapt it according to the egocentric sampling. The piecewise-smooth graphon model family overlaps with some low-rank models, but the two model classes are not nested. Its theory in egocentrically sampled and sparse networks is unclear. The computation is polynomial in $N$ but not well suited to networks with more than a few thousand nodes. We use the Python implementation provided from \citet{wu2018link} for it.
    \item ERGM \citep{handcock2010modeling}. We use the model version with egocentric sampling and geometrically weighted edgewise shared partnerships introduced by  \cite{hunter2007curved}. This ERGM configuration works best in our evaluation.  The model is implemented in the R package \textit{ergm} \citep{ergm}. Note that the ERGM does not follow the inhomogeneous Erd\"{o}s-Renyi framework in general. This method is the slowest among all the benchmarks, and it does not come with a known theoretical guarantee.
    \item SBM model fitting. This approach is adapted to egocentric sampling given the model fitting strategy of \cite{chen2018network} based on spectral clustering \citep{lei2015consistency}. It is implemented in the R package \textit{randnet} \citep{li2021randnet}. The model fitting is known to be consistent if the true model is the SBM. This model has a high computation speed similar to that of LE (the proposed method) and SE.
    \item DCBM model fitting. Similar to the SBM, the model fitting strategy of \cite{chen2018network} is applied based on the R package \textit{randnet}.
\end{enumerate}

\subsection{Experiment Setup}

Following \citet{wu2018link}, synthetic networks with dimension $N=500$ are generated under four network models.
\begin{enumerate}
    \item The SBM: In this model, we first randomly assign $N$ nodes to $K$ groups or communities. Let $\M{Z}\in\{0,1\}^{N\times K}$ be a community label matrix, where $Z_{ik} = 1$ if and only if the node $i$ belongs to the $k$th community. With a matrix symmetric matrix $\M{B} \in [0,1]^{K\times K}$ as the ``community-connection" matrix, the probability matrix is $\M{P} = \M{Z}\M{B}\M{Z}^T$. We set $K=5$ in our experiments. This model represents an important benchmark scenario because it is strictly low-rank ($K=5$) and we do have a model-based method to fit it consistently via the SBM fitting. Therefore, LE's performance under this model gives a measure of its adaptivity to specific models.
    \item The DCBM: The out-in ratio ($\beta$), i.e., the ratio of between-block to within-block edges, is set to be the ratio of the average between-block edge probability to within-block edge probability as in the "community-connection" matrix $B$ in the SBM. Similarly, the number of communities $K=5$. The degree parameters follow the power-law distribution with $\alpha=0.1$ and lower bound $1$. 
    \item The RDPG (product model). Under this model, we first generate $Z_i \in [0,1]^K$ for each node $i\in [N]$. Each coordinate of $Z_i$ is generated independently from $\mathrm{Beta}(0.5,1)$, following the procedure of \citet{athreya2017statistical}. The connection probability between nodes $i$ and $j$ is given by $P_{ij}=Z_i^\top Z_j$. This model provides a more general low-rank scenario compared to the SBM, for which no known consistency guarantees were previously available.
    \item The latent space model (distance model). For each node $i \in [N]$, we generate its latent vector $Z_i\sim N(0,I_5)$. Then the connection probability between nodes $i$ and $j$ is set by $P_{ij}=[1+\mathrm{exp}(\norm{Z_i-Z_j})]^{-1}$ where $\norm{Z_i-Z_j}$ is the Euclidean distance between $Z_i$ and $Z_j$. The $\M{P}$ of the distance model is not low-rank. This model thus serves to test the potential to approximate full-rank models.
\end{enumerate}
The network generated from the above models may not be sparse. Therefore, after deriving the probability matrix  $\M{P}$, we further scale it to control the average expected degree of the generated networks. We focus on networks in our experiments with expected degrees of $20$ and $50$. 

We evaluate the methods based on the mean squared error (MSE) on the unobserved probability matrix
\[\mathrm{MSE}=\norm{\hat{\M P}_{22}-\M P_{22}}_F^2/(N-n)^2.\]
Additionally, the performance is also measured by the timing of model fitting, reflecting the computational efficiency. Since most of the methods involve tuning procedures for which the configuration and preference can vary across users, we do not include the tuning procedure in timing evaluation and only focus on the model-fitting procedure. We are also aware that the different methods are implemented differently: NS is implemented in Python; LE, SE, SBM and DCBM are implemented in R; the major component of ERGM fitting is implemented in C and called by R functions. Overall, we believe the comparison of NS, SE, LE, and SBM is fair, while the ERGM, if implemented similarly, would be even more inferior in speed comparison. All experiments are independently repeated 100 times.

\subsection{Evaluation under Missing Completely At Random}

For egocentric missingness with missing-completely-at-random (MCAR),  we randomly sample $\rho = \{0.1,0.2,0.5,0.9\}$ of the nodes as observed while the rest as missing. We will investigate the performance of all methods in configurations of different synthetic models, expected degrees and $\rho$.

\subsubsection{Performance Evaluation}

\paragraph{Link Prediction Accuracy.}
Table~\ref{MSE-sim} displays the prediction errors under the SBM, DCBM, product, and distance models, respectively, with various levels of sampling proportions and sparsity levels.

Under the SBM model, the SBM model-fitting is generally the best method since it fits the correct model. LE, SE, LE+ and NS all have similar performances and are inferior to the SBM. LE is generally more accurate than SE, and its advantage increases with the sampling proportion $\rho$. The only case when SE is better is when $\rho$ is very low (0.1) while the network is dense (50). This coincides with our theoretical conjecture. LE+ can be better or similar to the better one of the two.

Under the DCBM model, LE+ is generally the best method, while LE and SE have similar performance. The comparison between LE, SE and LE+ remains similar. A surprising result is that DCBM model fitting does not produce the best result and is inferior to  SBM model fitting under low-degree sparse networks.

Under the product and distance models, LE and LE+ are essentially the best ones. These experiments demonstrate the effectiveness of the proposed estimator to fit general low-rank structures and even beyond that.

One interesting fact may be that SBM performs well under sparse networks (Deg.$=20$) regardless of the generating model. This may be explained by its parsimony, reducing the variance in estimation when little information is available from the observed data.

\begin{table}[!ht]
\centering
\caption{\label{MSE-sim}MSE of Link Prediction Performance on Synthetic Networks ($10^{-3}$).}
\resizebox{0.65\textwidth}{!}{
\begin{tabular}{| c | c | c c c c c c c |}
\hline
Model&($\rho$,Deg.)&LE&SE&LE+&NS&ERGM&SBM&DCBM\\
\hline
\multirow{16}{*}{SBM}&\multirow{2}{*}{(0.1,20)}&4.81&5.33&\textbf{4.19}&7.79&9.15&6.43&8.16\\
&&(0.139)&(0.228)&(0.101)&(0.098)&(0.056)&(0.467)&(0.37)\\
\cline{2-9}
&\multirow{2}{*}{(0.1,50)}&14.4&12.5&11.2&17.7&43.1&\textbf{6.71}&13.8\\
&&(0.246)&(0.236)&(0.149)&(0.285)&(0.1)&(0.161)&(0.168)\\
\cline{2-9}
&\multirow{2}{*}{(0.2,20)}&2.87&2.88&2.44&4.51&9.23&\textbf{2.41}&3.26\\
&&(0.043)&(0.049)&(0.031)&(0.037)&(0.036)&(0.043)&(0.046)\\
\cline{2-9}
&\multirow{2}{*}{(0.2,50)}&7.44&7.6&6.96&10.2&42.8&\textbf{3.03}&8.25\\
&&(0.105)&(0.087)&(0.069)&(0.165)&(0.069)&(0.054)&(0.133)\\
\cline{2-9}
&\multirow{2}{*}{(0.5,20)}&1.35&1.52&1.35&2.04&9.8&\textbf{0.763}&1.96\\
&&(0.016)&(0.017)&(0.014)&(0.021)&(0.026)&(0.02)&(0.021)\\
\cline{2-9}
&\multirow{2}{*}{(0.5,50)}&3.54&4.02&3.64&5.1&43.4&\textbf{0.823}&3.4\\
&&(0.027)&(0.05)&(0.035)&(0.087)&(0.129)&(0.039)&(0.077)\\
\cline{2-9}
&\multirow{2}{*}{(0.9,20)}&1.05&1.43&1.27&1.24&10.5&\textbf{0.548}&1.81\\
&&(0.018)&(0.041)&(0.033)&(0.021)&(0.108)&(0.027)&(0.047)\\
\cline{2-9}
&\multirow{2}{*}{(0.9,50)}&2.18&3.82&2.87&1.78&45.5&\textbf{0.584}&2.6\\
&&(0.026)&(0.08)&(0.044)&(0.035)&(0.561)&(0.043)&(0.078)\\
\hline
\multirow{16}{*}{DCBM}&\multirow{2}{*}{(0.1,20)}&7.2&8.38&\textbf{6.21}&8.46&10.2&9.84&10.4\\
&&(0.14)&(0.231)&(0.138)&(0.084)&(0.051)&(0.67)&(0.495)\\
\cline{2-9}
&\multirow{2}{*}{(0.1,50)}&20.4&16.4&14.7&21.9&50.2&\textbf{13.6}&16.3\\
&&(0.229)&(0.208)&(0.131)&(0.172)&(0.112)&(0.183)&(0.225)\\
\cline{2-9}
&\multirow{2}{*}{(0.2,20)}&4.37&4.01&\textbf{3.43}&5.06&10.3&3.67&4.23\\
&&(0.043)&(0.045)&(0.028)&(0.04)&(0.034)&(0.065)&(0.043)\\
\cline{2-9}
&\multirow{2}{*}{(0.2,50)}&9.12&8.67&8.1&14.6&50.3&\textbf{7.51}&7.91\\
&&(0.139)&(0.076)&(0.079)&(0.135)&(0.055)&(0.183)&(0.152)\\
\cline{2-9}
&\multirow{2}{*}{(0.5,20)}&\textbf{1.68}&1.98&1.69&2.76&11.1&1.43&1.98\\
&&(0.022)&(0.024)&(0.018)&(0.023)&(0.032)&(0.032)&(0.027)\\
\cline{2-9}
&\multirow{2}{*}{(0.5,50)}&3.39&3.69&3.46&8.93&50.6&5.95&\textbf{3.08}\\
&&(0.028)&(0.035)&(0.036)&(0.107)&(0.175)&(0.093)&(0.148)\\
\cline{2-9}
&\multirow{2}{*}{(0.9,20)}&\textbf{1.04}&1.44&1.23&1.93&13.1&1.21&1.43\\
&&(0.023)&(0.043)&(0.028)&(0.051)&(0.16)&(0.06)&(0.07)\\
\cline{2-9}
&\multirow{2}{*}{(0.9,50)}&\textbf{2.25}&2.65&2.53&4.68&54.9&7.73&2.82\\
&&(0.052)&(0.061)&(0.051)&(0.196)&(0.72)&(0.356)&(0.098)\\
\hline
\multirow{16}{*}{product}&\multirow{2}{*}{(0.1,20)}&3&3.16&\textbf{2.37}&6.61&7.15&8.48&6.96\\
&&(0.136)&(0.125)&(0.085)&(0.079)&(0.058)&(0.646)&(0.297)\\
\cline{2-9}
&\multirow{2}{*}{(0.1,50)}&5.63&5.65&\textbf{5.12}&13.6&33.1&6.49&6.75\\
&&(0.049)&(0.084)&(0.049)&(0.126)&(0.181)&(0.213)&(0.104)\\
\cline{2-9}
&\multirow{2}{*}{(0.2,20)}&1.17&1.23&\textbf{1.06}&3.44&7.37&1.77&2.03\\
&&(0.011)&(0.019)&(0.011)&(0.025)&(0.041)&(0.105)&(0.049)\\
\cline{2-9}
&\multirow{2}{*}{(0.2,50)}&3.33&3.39&\textbf{3.24}&7.84&33.1&4.7&4.16\\
&&(0.02)&(0.033)&(0.024)&(0.046)&(0.103)&(0.029)&(0.031)\\
\cline{2-9}
&\multirow{2}{*}{(0.5,20)}&\textbf{0.582}&0.67&0.592&1.54&7.44&0.855&0.968\\
&&(0.005)&(0.01)&(0.006)&(0.01)&(0.029)&(0.01)&(0.013)\\
\cline{2-9}
&\multirow{2}{*}{(0.5,50)}&\textbf{1.99}&2.09&\textbf{2}&4.74&33.2&3.86&3.06\\
&&(0.013)&(0.021)&(0.017)&(0.038)&(0.065)&(0.021)&(0.023)\\
\cline{2-9}
&\multirow{2}{*}{(0.9,20)}&\textbf{0.413}&0.665&0.478&1.09&7.44&0.988&1.03\\
&&(0.007)&(0.021)&(0.008)&(0.014)&(0.072)&(0.027)&(0.033)\\
\cline{2-9}
&\multirow{2}{*}{(0.9,50)}&\textbf{1.6}&2&1.73&2.9&33.8&3.82&3.07\\
&&(0.024)&(0.041)&(0.026)&(0.042)&(0.275)&(0.062)&(0.055)\\
\hline
\multirow{16}{*}{distance}&\multirow{2}{*}{(0.1,20)}&3.85&4.53&\textbf{3.38}&7.2&7.82&9.01&8.32\\
&&(0.137)&(0.213)&(0.124)&(0.086)&(0.061)&(0.651)&(0.378)\\
\cline{2-9}
&\multirow{2}{*}{(0.1,50)}&9.87&9.62&\textbf{9.08}&17.5&36.4&10.1&10.7\\
&&(0.078)&(0.092)&(0.06)&(0.124)&(0.15)&(0.203)&(0.124)\\
\cline{2-9}
&\multirow{2}{*}{(0.2,20)}&1.87&1.86&\textbf{1.73}&4.19&7.86&2.22&2.61\\
&&(0.017)&(0.02)&(0.014)&(0.03)&(0.041)&(0.05)&(0.035)\\
\cline{2-9}
&\multirow{2}{*}{(0.2,50)}&7.49&7.45&\textbf{7.18}&11.6&36.7&8.47&8.3\\
&&(0.028)&(0.045)&(0.026)&(0.051)&(0.093)&(0.03)&(0.044)\\
\cline{2-9}
&\multirow{2}{*}{(0.5,20)}&1.34&1.38&\textbf{1.31}&2.26&8.05&1.53&1.73\\
&&(0.005)&(0.008)&(0.006)&(0.011)&(0.027)&(0.01)&(0.012)\\
\cline{2-9}
&\multirow{2}{*}{(0.5,50)}&6.22&6.34&\textbf{5.8}&8.72&37.7&8.17&7.35\\
&&(0.026)&(0.026)&(0.02)&(0.041)&(0.054)&(0.031)&(0.029)\\
\cline{2-9}
&\multirow{2}{*}{(0.9,20)}&\textbf{2.13}&2.23&2.18&2.78&9.69&2.53&2.61\\
&&(0.011)&(0.02)&(0.016)&(0.016)&(0.07)&(0.017)&(0.017)\\
\cline{2-9}
&\multirow{2}{*}{(0.9,50)}&9.49&9.5&\textbf{9.1}&12.7&46.1&13.4&12\\
&&(0.049)&(0.054)&(0.047)&(0.051)&(0.266)&(0.073)&(0.062)\\
\hline
\end{tabular}
}
\end{table}

\paragraph{Timing Comparison.} Table~\ref{runtime-sim} summarizes the average computing time of all the methods under evaluation. The ERGM is by far the slowest and generally not applicable for networks with over a few hundred nodes. NS is feasible but is still slower than the other methods. The comparison between LE, SE, SBM and DCBM can be different across network configurations, but overall, they are all comparably efficient in speed.

\begin{table}[!ht]
\centering
\caption{\label{runtime-sim}Timing Comparison between Benchmark Methods on Synthetic Networks (In Milliseconds).}
\resizebox{0.65\textwidth}{!}{
\begin{tabular}{| c | c | c c c c c c c |}
\hline
Model&($\rho$,Deg.)&LE&SE&LE+&NS&ERGM&SBM&DCBM\\
\hline
\multirow{16}{*}{SBM}&\multirow{2}{*}{(0.1,20)}&7.09&9.39&9.32&134&27000&36.9&46\\
&&(2.5)&(2.97)&(1.71)&(0.176)&(147)&(1.1)&(2)\\
\cline{2-9}
&\multirow{2}{*}{(0.1,50)}&3.78&7.55&11.2&133&35000&13.2&31.3\\
&&(0.108)&(2.4)&(2.42)&(0.0811)&(248)&(0.387)&(0.744)\\
\cline{2-9}
&\multirow{2}{*}{(0.2,20)}&10.9&7.55&12&150&21400&37.1&49.4\\
&&(3.07)&(1.63)&(0.144)&(0.0346)&(127)&(1.1)&(2.09)\\
\cline{2-9}
&\multirow{2}{*}{(0.2,50)}&9.66&11.8&15.2&151&30100&13.6&31.7\\
&&(2.62)&(3.08)&(1.65)&(0.0678)&(183)&(1.57)&(0.913)\\
\cline{2-9}
&\multirow{2}{*}{(0.5,20)}&24.3&28.6&46.9&197&11400&19.1&52.2\\
&&(2.9)&(4.05)&(2.24)&(0.0479)&(114)&(1.52)&(1.14)\\
\cline{2-9}
&\multirow{2}{*}{(0.5,50)}&26.8&25.6&50.8&197&21400&19.8&25.6\\
&&(2.98)&(3.08)&(1.99)&(0.0921)&(168)&(2.19)&(1.63)\\
\cline{2-9}
&\multirow{2}{*}{(0.9,20)}&78.2&72.6&197&307&8080&25.2&79.2\\
&&(3.38)&(2.88)&(3.84)&(0.706)&(88.4)&(1.59)&(1.44)\\
\cline{2-9}
&\multirow{2}{*}{(0.9,50)}&72.6&79.2&212&311&18700&28.8&36.5\\
&&(1.76)&(3.67)&(4.77)&(0.109)&(193)&(0.69)&(0.915)\\
\hline
\multirow{16}{*}{DCBM}&\multirow{2}{*}{(0.1,20)}&7.56&4.75&11.9&139&26700&46.5&48.9\\
&&(2.53)&(0.0665)&(2.63)&(0.261)&(133)&(2.04)&(0.981)\\
\cline{2-9}
&\multirow{2}{*}{(0.1,50)}&10.1&5.81&9.23&138&35600&18.1&34.5\\
&&(3.31)&(1.73)&(1.77)&(0.109)&(216)&(0.385)&(0.811)\\
\cline{2-9}
&\multirow{2}{*}{(0.2,20)}&8.16&6.71&18.1&148&21200&43.7&52.1\\
&&(1.73)&(0.141)&(2.92)&(0.219)&(129)&(1.52)&(2.2)\\
\cline{2-9}
&\multirow{2}{*}{(0.2,50)}&5.75&7.57&16.6&147&30500&17.4&36.6\\
&&(0.115)&(1.76)&(2.9)&(0.166)&(207)&(0.281)&(1.68)\\
\cline{2-9}
&\multirow{2}{*}{(0.5,20)}&20.9&24.8&44.6&185&11200&21.2&49.5\\
&&(2.78)&(3.23)&(1.67)&(0.257)&(121)&(0.3)&(1.96)\\
\cline{2-9}
&\multirow{2}{*}{(0.5,50)}&20.9&18&51.2&185&21600&26&27.8\\
&&(2.87)&(0.136)&(3.81)&(0.202)&(162)&(1.52)&(0.732)\\
\cline{2-9}
&\multirow{2}{*}{(0.9,20)}&67.3&68.9&185&251&8180&22&68.5\\
&&(2.44)&(2.52)&(3)&(0.485)&(91.3)&(1.58)&(2.99)\\
\cline{2-9}
&\multirow{2}{*}{(0.9,50)}&68.9&72.5&196&249&18900&64.4&43.4\\
&&(1.31)&(2.97)&(4.17)&(0.254)&(175)&(1.82)&(2.38)\\
\hline
\multirow{16}{*}{product}&\multirow{2}{*}{(0.1,20)}&8.57&5.33&7.22&136&26900&48.6&55.8\\
&&(3.03)&(1.57)&(0.0859)&(0.158)&(143)&(1.07)&(2.09)\\
\cline{2-9}
&\multirow{2}{*}{(0.1,50)}&3.25&8.62&12.1&137&35700&37.2&42.7\\
&&(0.0948)&(2.98)&(2.94)&(0.178)&(264)&(0.875)&(3.35)\\
\cline{2-9}
&\multirow{2}{*}{(0.2,20)}&10.5&7.26&11.3&147&21600&54.6&59.9\\
&&(3.08)&(1.68)&(0.146)&(0.186)&(127)&(1.39)&(2.33)\\
\cline{2-9}
&\multirow{2}{*}{(0.2,50)}&6.7&8.75&14.6&147&30800&54&42.2\\
&&(1.71)&(2.43)&(2.55)&(0.12)&(227)&(0.635)&(0.906)\\
\cline{2-9}
&\multirow{2}{*}{(0.5,20)}&23.1&26.2&44.9&180&11100&63.5&74.2\\
&&(2.96)&(3.81)&(2.36)&(0.431)&(108)&(1.77)&(2.31)\\
\cline{2-9}
&\multirow{2}{*}{(0.5,50)}&19.5&17.9&44.9&180&21900&80.8&69.2\\
&&(1.74)&(1.65)&(2.38)&(0.303)&(172)&(0.668)&(2.95)\\
\cline{2-9}
&\multirow{2}{*}{(0.9,20)}&66.6&65.8&188&247&7690&102&98.6\\
&&(2.65)&(2.81)&(3.16)&(0.393)&(94.3)&(1.98)&(2.67)\\
\cline{2-9}
&\multirow{2}{*}{(0.9,50)}&66.3&64.2&187&246&19200&93.6&76.6\\
&&(2.28)&(2.1)&(2.96)&(0.318)&(170)&(0.613)&(2.34)\\
\hline
\multirow{16}{*}{distance}&\multirow{2}{*}{(0.1,20)}&4.85&5.9&7.03&138&27000&55&55.1\\
&&(1.47)&(1.74)&(0.164)&(0.249)&(133)&(2.12)&(0.827)\\
\cline{2-9}
&\multirow{2}{*}{(0.1,50)}&3.22&8.28&7.65&137&35600&42.8&36.7\\
&&(0.0602)&(2.66)&(1.49)&(0.095)&(226)&(0.901)&(0.814)\\
\cline{2-9}
&\multirow{2}{*}{(0.2,20)}&5.14&7.49&16.2&147&21500&64.9&58\\
&&(0.0914)&(1.68)&(2.93)&(0.229)&(142)&(2.75)&(2.06)\\
\cline{2-9}
&\multirow{2}{*}{(0.2,50)}&5.05&5.45&16.2&148&30600&61.4&44.3\\
&&(0.0567)&(0.0868)&(3.1)&(0.147)&(240)&(0.788)&(0.976)\\
\cline{2-9}
&\multirow{2}{*}{(0.5,20)}&18&18.7&46.4&183&11100&69&68.5\\
&&(0.129)&(1.61)&(3.67)&(0.37)&(103)&(2.52)&(1.26)\\
\cline{2-9}
&\multirow{2}{*}{(0.5,50)}&22.5&20.4&45.1&183&21800&91.1&67.9\\
&&(2.56)&(2.27)&(3.47)&(0.295)&(166)&(0.851)&(2.07)\\
\cline{2-9}
&\multirow{2}{*}{(0.9,20)}&68.3&67.1&164&250&7700&110&96.5\\
&&(2.65)&(0.687)&(5.2)&(0.561)&(91.2)&(2.73)&(1.53)\\
\cline{2-9}
&\multirow{2}{*}{(0.9,50)}&68.9&74.4&158&246&19100&108&72.6\\
&&(2.47)&(3.76)&(2.9)&(0.283)&(174)&(0.75)&(1.94)\\
\hline
\end{tabular}
}
\end{table}

\paragraph{Summary.} Overall, the LE method renders competitive accuracy in all settings even when the model is full-rank (but can be approximated by low-rank structures). The closely related SE method is inferior to LE, except in low-sampling-high-density cases. LE+ can generally achieve adaptive performance by combining LE and SE. 

\subsection{Evaluation under Missing Not At Random}

All methods are also evaluated under the missing-not-at-random (MNAR) scenario, where the sampling probability of a node depends on its degree (in $\M{A}$). Two settings are evaluated: the positive setting ($+$), where nodes with higher degrees are more likely to be sampled; and the negative setting ($-$), where nodes with lower degrees are more likely to be sampled. Based on the adjacency matrix $A$, the nodes are divided into three groups  by the descending order of their node degrees by proportions $33\%,34\%$ and $33\%$. The three groups are sampled at $\delta_+=\{1.5,1,0.5\}$ or $\delta_-=\{0.5,1,1.5\}$ respectively of $\rho=\{0.2,0.5\}$. For example, the first group under $\rho=0.2$ and the positive setting is sampled at $\tilde{\rho}=\delta_+\rho=1.5\times0.2=0.3$.


Tables~\ref{MSE-nonuniform}--\ref{MSE-nonuniinv} display the MSE of the prediction error under the four generating models, with varying sampling proportions and sparsity levels for both positive and negative settings. The result is consistent with the findings under uniform sampling. Under the SBM model, SBM still has the best general performance. Under other settings, LE+ is generally the best method while LE and SE have comparable performance. Timing comparison is not included as it is similar to that of uniform sampling.

\begin{table}[!ht]
\centering
\caption{\label{MSE-nonuniform}MSE of Link Prediction Performance on Synthetic Networks, Positive MNAR Setting ($10^{-3}$).}
\resizebox{0.65\textwidth}{!}{
\begin{tabular}{| c | c | c c c c c c c |}
\hline
Model&($\rho$,Deg.)&LE&SE&LE+&NS&ERGM&SBM&DCBM\\
\hline
\multirow{8}{*}{SBM}&\multirow{2}{*}{(0.2,20)}&2.05&2.14&1.89&3.48&10.4&\textbf{1.72}&2.54\\
&&(0.023)&(0.023)&(0.016)&(0.026)&(0.032)&(0.026)&(0.03)\\
\cline{2-9}
&\multirow{2}{*}{(0.2,50)}&6.41&6.11&5.85&6.68&41.7&\textbf{2.49}&6.37\\
&&(0.083)&(0.061)&(0.048)&(0.062)&(0.13)&(0.035)&(0.117)\\
\cline{2-9}
&\multirow{2}{*}{(0.5,20)}&1.12&1.06&1.04&1.37&9.66&\textbf{0.549}&1.14\\
&&(0.014)&(0.01)&(0.009)&(0.008)&(0.033)&(0.009)&(0.012)\\
\cline{2-9}
&\multirow{2}{*}{(0.5,50)}&3.88&4.19&3.78&3.82&37&\textbf{1.26}&5.94\\
&&(0.015)&(0.045)&(0.023)&(0.028)&(0.082)&(0.038)&(0.12)\\
\hline
\multirow{8}{*}{DCBM}&\multirow{2}{*}{(0.2,20)}&3.6&3.39&\textbf{2.9}&4.4&11.4&2.97&3.63\\
&&(0.036)&(0.038)&(0.025)&(0.024)&(0.032)&(0.032)&(0.045)\\
\cline{2-9}
&\multirow{2}{*}{(0.2,50)}&7.26&7.2&6.92&12&51.3&6.61&\textbf{6.24}\\
&&(0.103)&(0.055)&(0.048)&(0.087)&(0.069)&(0.123)&(0.131)\\
\cline{2-9}
&\multirow{2}{*}{(0.5,20)}&1.12&1.26&1.13&2&11.1&\textbf{1.07}&1.16\\
&&(0.014)&(0.012)&(0.01)&(0.016)&(0.039)&(0.019)&(0.02)\\
\cline{2-9}
&\multirow{2}{*}{(0.5,50)}&2.4&2.6&2.45&6.53&44.9&4.5&\textbf{2.17}\\
&&(0.018)&(0.024)&(0.021)&(0.068)&(0.105)&(0.049)&(0.125)\\
\hline
\multirow{8}{*}{product}&\multirow{2}{*}{(0.2,20)}&0.93&0.961&\textbf{0.87}&2.91&8.94&1.39&1.46\\
&&(0.006)&(0.012)&(0.008)&(0.021)&(0.032)&(0.05)&(0.023)\\
\cline{2-9}
&\multirow{2}{*}{(0.2,50)}&2.8&2.83&\textbf{2.71}&6.33&36.9&4.77&3.9\\
&&(0.013)&(0.026)&(0.015)&(0.021)&(0.073)&(0.033)&(0.037)\\
\cline{2-9}
&\multirow{2}{*}{(0.5,20)}&\textbf{0.392}&0.43&\textbf{0.392}&1&8.39&0.702&0.58\\
&&(0.003)&(0.005)&(0.003)&(0.005)&(0.028)&(0.006)&(0.006)\\
\cline{2-9}
&\multirow{2}{*}{(0.5,50)}&\textbf{1.4}&1.46&\textbf{1.4}&3.19&35.2&3.4&2.16\\
&&(0.007)&(0.013)&(0.008)&(0.019)&(0.066)&(0.022)&(0.015)\\
\hline
\multirow{8}{*}{distance}&\multirow{2}{*}{(0.2,20)}&1.58&1.57&\textbf{1.51}&3.52&9.35&1.93&2.07\\
&&(0.008)&(0.011)&(0.007)&(0.022)&(0.028)&(0.048)&(0.022)\\
\cline{2-9}
&\multirow{2}{*}{(0.2,50)}&6.8&6.74&\textbf{6.54}&9.92&40.2&8.26&7.7\\
&&(0.02)&(0.026)&(0.017)&(0.02)&(0.073)&(0.026)&(0.044)\\
\cline{2-9}
&\multirow{2}{*}{(0.5,20)}&\textbf{1.13}&1.15&\textbf{1.12}&1.73&8.94&1.31&1.3\\
&&(0.003)&(0.006)&(0.003)&(0.004)&(0.029)&(0.006)&(0.006)\\
\cline{2-9}
&\multirow{2}{*}{(0.5,50)}&5.5&5.44&\textbf{5.17}&7.35&38.7&7.08&6.15\\
&&(0.017)&(0.016)&(0.016)&(0.02)&(0.067)&(0.019)&(0.015)\\
\hline
\end{tabular}
}
\end{table}
\begin{table}[!ht]
    \centering
    \caption{\label{MSE-nonuniinv}MSE of Link Prediction Performance on Synthetic Networks, Negative MNAR Setting ($10^{-3}$).}
\resizebox{0.65\textwidth}{!}{
\begin{tabular}{| c | c | c c c c c c c |}
\hline
Model&($\rho$,Deg.)&LE&SE&LE+&NS&ERGM&SBM&DCBM\\
\hline
\multirow{8}{*}{SBM}&\multirow{2}{*}{(0.2,20)}&4.13&4.1&\textbf{3.37}&5.59&8.16&3.78&4.26\\
&&(0.039)&(0.083)&(0.044)&(0.03)&(0.02)&(0.109)&(0.069)\\
\cline{2-9}
&\multirow{2}{*}{(0.2,50)}&11.3&9.5&9.1&18&42.5&\textbf{4.85}&9.89\\
&&(0.198)&(0.095)&(0.093)&(0.139)&(0.056)&(0.121)&(0.158)\\
\cline{2-9}
&\multirow{2}{*}{(0.5,20)}&2.55&2.75&2.33&4.39&10.1&\textbf{1.75}&3.08\\
&&(0.03)&(0.041)&(0.027)&(0.036)&(0.025)&(0.048)&(0.028)\\
\cline{2-9}
&\multirow{2}{*}{(0.5,50)}&5.49&6.23&5.59&18.8&53.9&\textbf{0.638}&5.46\\
&&(0.057)&(0.069)&(0.048)&(0.183)&(0.089)&(0.047)&(0.124)\\
\hline
\multirow{8}{*}{DCBM}&\multirow{2}{*}{(0.2,20)}&5.02&4.72&\textbf{3.95}&5.85&9.33&4.23&5.01\\
&&(0.04)&(0.052)&(0.029)&(0.03)&(0.022)&(0.072)&(0.071)\\
\cline{2-9}
&\multirow{2}{*}{(0.2,50)}&11.2&9.95&9.41&18&49.2&\textbf{9.35}&9.64\\
&&(0.144)&(0.093)&(0.075)&(0.1)&(0.054)&(0.199)&(0.214)\\
\cline{2-9}
&\multirow{2}{*}{(0.5,20)}&2.49&2.87&2.38&4.02&11.3&\textbf{2.21}&3.04\\
&&(0.027)&(0.029)&(0.02)&(0.026)&(0.032)&(0.039)&(0.048)\\
\cline{2-9}
&\multirow{2}{*}{(0.5,50)}&\textbf{4.7}&5.32&4.79&15.4&57.9&9.66&5.13\\
&&(0.034)&(0.05)&(0.035)&(0.118)&(0.163)&(0.106)&(0.27)\\
\hline
\multirow{8}{*}{product}&\multirow{2}{*}{(0.2,20)}&1.5&1.54&\textbf{1.35}&3.92&5.82&3.24&2.61\\
&&(0.012)&(0.025)&(0.017)&(0.025)&(0.022)&(0.217)&(0.064)\\
\cline{2-9}
&\multirow{2}{*}{(0.2,50)}&4.2&4.03&\textbf{3.85}&9.84&29&5.14&4.89\\
&&(0.022)&(0.03)&(0.024)&(0.034)&(0.057)&(0.051)&(0.046)\\
\cline{2-9}
&\multirow{2}{*}{(0.5,20)}&0.866&1&\textbf{0.861}&2.24&6.46&1.33&1.49\\
&&(0.006)&(0.013)&(0.008)&(0.011)&(0.021)&(0.029)&(0.015)\\
\cline{2-9}
&\multirow{2}{*}{(0.5,50)}&2.78&2.88&\textbf{2.76}&8.58&32.1&5.32&4.16\\
&&(0.015)&(0.025)&(0.018)&(0.033)&(0.056)&(0.036)&(0.03)\\
\hline
\multirow{8}{*}{distance}&\multirow{2}{*}{(0.2,20)}&2.3&2.37&\textbf{2.06}&4.53&6.48&3.31&3.3\\
&&(0.023)&(0.044)&(0.018)&(0.026)&(0.021)&(0.146)&(0.068)\\
\cline{2-9}
&\multirow{2}{*}{(0.2,50)}&8.33&8.27&\textbf{7.94}&13.5&32.9&9.01&9.13\\
&&(0.029)&(0.047)&(0.025)&(0.034)&(0.051)&(0.043)&(0.046)\\
\cline{2-9}
&\multirow{2}{*}{(0.5,20)}&1.7&1.79&\textbf{1.68}&3.01&7.34&1.98&2.31\\
&&(0.008)&(0.014)&(0.009)&(0.009)&(0.021)&(0.017)&(0.017)\\
\cline{2-9}
&\multirow{2}{*}{(0.5,50)}&7.55&7.74&\textbf{6.96}&13.3&37.5&10.5&9.08\\
&&(0.03)&(0.035)&(0.026)&(0.037)&(0.05)&(0.038)&(0.028)\\
\hline
\end{tabular}
}
\end{table}

\paragraph{Summary.} Overall, no extreme change in performance is observed across the tested missingness mechanisms. This indicates that non-uniform sampling does not severely affect our method LE and any benchmark methods, supporting our theoretical claim that our method is applicable for non-uniform sampling as long as \ref{ass:rank} holds. 

\section{LINK PREDICTION IN REAL-WORLD NETWORKS}

In this section, we evaluate our approach to link prediction on real-world networks. We consider three examples: two social networks and one airline traffic network. We wish to demonstrate the merits and limitations of our approach.  In the first two examples, our method outperforms other benchmarks, indicating that the low-rank model assumption is reasonable. In the third example, the NS method works better. As such, the network may not be well-approximated with a low-rank structure, whereas the graphon structure underlying the NS method might be more suitable.

\subsection{Data Information}

The first network is the Enron email network of \citet{priebe2005scan} between 184 employees of the Enron company; edges indicate employees’ email communication. The second network is a faculty friendship network between 81 faculty members at a UK university  \citep{nepusz2008fuzzy}. The last network contains 755 airports in the United States, based on the U.S. Bureau of Transportation Statistics, where two airports are connected if there is a direct passenger flight between them. For simplicity, we ignore edge directions and weights. 
The average node density, betweenness centrality, and closeness centrality of each node are shown in  Figure \ref{Real-attr}. Overall, the airport network exhibits much stronger heterogeneity in topological features. The nodes are, on average, much less connected and much less central, while special hub airports provide highly dense and central connections. Such strong topological heterogeneity suggests that the low-rank model may not be able to approximate this network well, as might occur in real-world situations  \citep{seshadhri2020impossibility}. As demonstrated in the link prediction evaluation, this topological variation does lead to a preference for different link prediction strategies.

\subsection{Performance Evaluation}

Contrary to our synthetic experiments, it is impossible to evaluate the MSE of the missing probability matrix.
Instead, we specifically evaluate the methods based on their link prediction accuracy on the unobserved entries ($\M{A}_{22}$). Because the entries in $\M{A}_{22}$ are binary, for a given threshold on values of $\hat{\M{P}}_{22}$, we get
\begin{align*}
\mathrm{TPR}&=\frac{\#\{\text{Correctly predicted edges}\}}{\#\{\text{Total existing edges}\}}\\
\mathrm{FPR}&=\frac{\#\{\text{Incorrectly predicted edges}\}}{\#\{\text{Total existing non-edges}\}}.
\end{align*}
Varying the threshold and plotting the true positive rate (TPR) against the false positive rate (FPR) produces a ROC curve. We evaluate the link prediction accuracy using the area under the ROC curves (AUC) as our metric for link prediction performance. ROC curve is a commonly used performance metric in link prediction problems \citep{liben2007link,zhao2017link}.
\begin{figure}[!ht]
    \centering
    \includegraphics[width=0.75\textwidth]{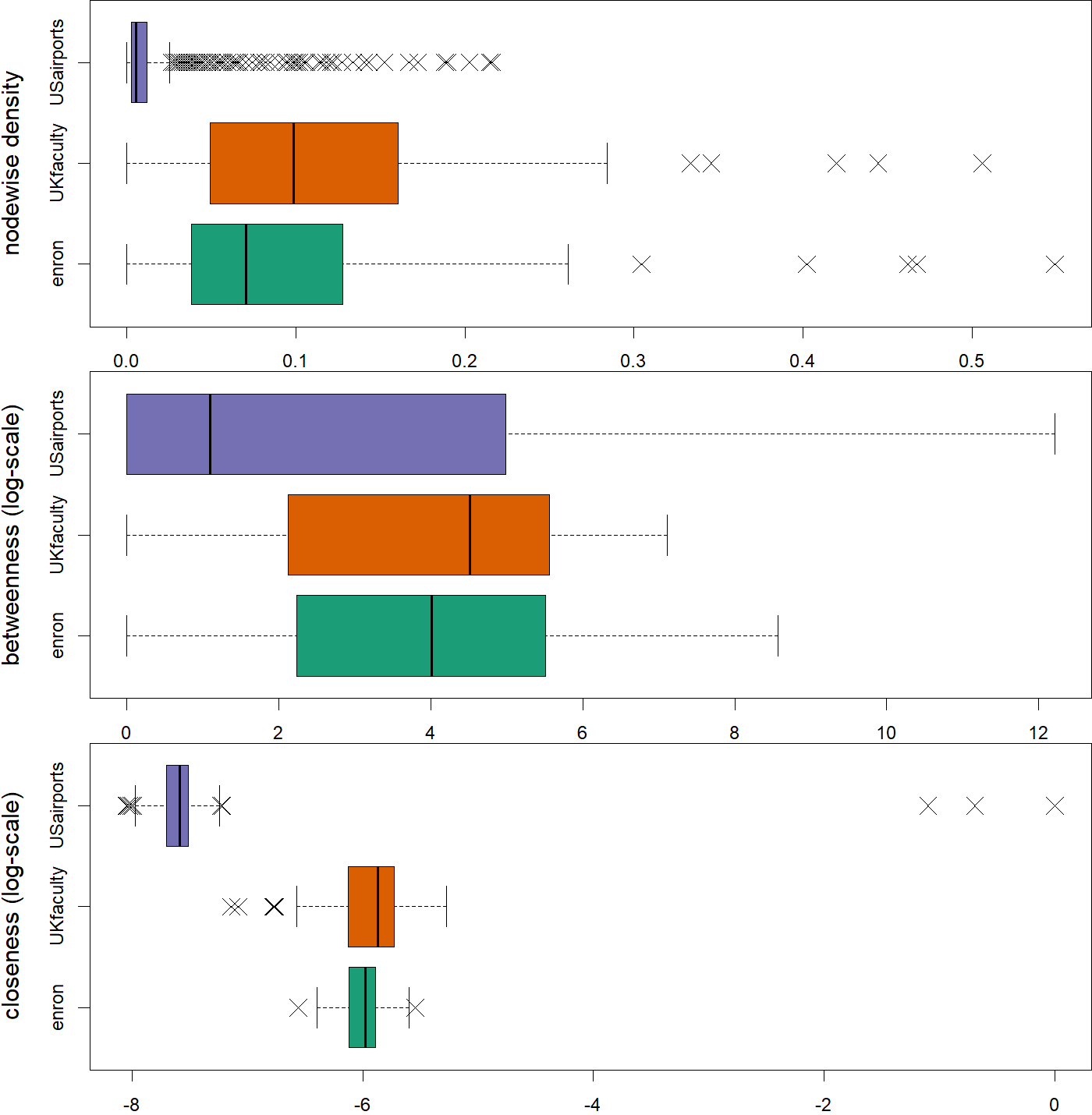}
    \caption{\label{Real-attr}Attributes of Real Networks; Node Degree (Left), Betweenness (Middle), and Closeness (Right).}
\end{figure}

\subsubsection{Evaluation under MAR}

The same sampling mechanism as for synthetic networks is applied. Results for the three networks are summarized in Table \ref{AUC-real}. For the Enron and UK faculty networks, LE+ has the best prediction performance among all methods. Meanwhile, LE has comparable performance to LE+ and is better than other methods. The low-rank model is thus suitable for these two networks. The NS method also delivers reasonably good performance thanks to its generality. For the U.S. airport network, under low sampling proportions $\rho=\{0.1,0.2\}$, LE, SE, LE+, NS and DCBM are similar and are inferior to the SBM method. Under higher sampling proportions $\rho=\{0.5,0.9\}$, the NS method has the best performance, while the other benchmark methods perform similarly. This pattern implies that all low-rank methods perform no better than a block approximation while missing additional structures which the NS method might integrate. Such an observation conveys that low-rank models may not be feasible for this airport network. Also, this is consistent with the synthetic experiments where SBM performs the best under very low sampling proportions due to its simplicity. The missing values for the UK faculty network under the sampling proportion $\rho=0.1$ is not produced because it has a small network size (\(N=81\)) inappropriate for such a low value of $\rho$. 

\begin{table}[!ht]
\centering
\caption{\label{AUC-real}Predictive AUC on Three Networks.}
\resizebox{0.65\textwidth}{!}{
\begin{tabular}{| c | c | c c c c c c c |}
\hline
Dataset&$\rho$&LE&SE&LE+&NS&ERGM&SBM&DCBM\\
\hline
\multirow{8}{*}{enron}&\multirow{2}{*}{0.1}&0.7&0.699&\textbf{0.724}&0.634&0.508&0.663&0.68\\
&&(0.004)&(0.005)&(0.005)&(0.004)&(0.003)&(0.003)&(0.004)\\
\cline{2-9}
&\multirow{2}{*}{0.2}&0.784&0.774&\textbf{0.803}&0.718&0.508&0.713&0.741\\
&&(0.003)&(0.003)&(0.004)&(0.004)&(0.003)&(0.004)&(0.004)\\
\cline{2-9}
&\multirow{2}{*}{0.5}&0.874&0.85&\textbf{0.882}&0.824&0.527&0.774&0.788\\
&&(0.003)&(0.002)&(0.002)&(0.003)&(0.004)&(0.003)&(0.003)\\
\cline{2-9}
&\multirow{2}{*}{0.9}&0.903&0.882&\textbf{0.909}&0.845&0.64&0.805&0.828\\
&&(0.005)&(0.005)&(0.004)&(0.006)&(0.009)&(0.009)&(0.007)\\
\hline
\multirow{7}{*}{UK faculty}&\multirow{2}{*}{0.2}&0.728&0.701&\textbf{0.731}&0.631&0.51&0.657&0.667\\
&&(0.007)&(0.006)&(0.009)&(0.007)&(0.002)&(0.008)&(0.006)\\
\cline{2-9}
&\multirow{2}{*}{0.5}&0.837&0.792&\textbf{0.842}&0.77&0.516&0.724&0.739\\
&&(0.004)&(0.005)&(0.004)&(0.005)&(0.003)&(0.007)&(0.007)\\
\cline{2-9}
&\multirow{2}{*}{0.9}&\textbf{0.845}&0.768&\textbf{0.846}&0.785&0.606&0.733&0.747\\
&&(0.012)&(0.016)&(0.012)&(0.016)&(0.014)&(0.02)&(0.018)\\
\hline
\multirow{8}{*}{US airports}&\multirow{2}{*}{0.1}&0.76&0.776&0.775&0.742&0.502&\textbf{0.835}&0.798\\
&&(0.006)&(0.006)&(0.006)&(0.005)&(<0.001)&(0.002)&(0.003)\\
\cline{2-9}
&\multirow{2}{*}{0.2}&0.822&0.82&0.837&0.817&0.505&\textbf{0.859}&0.837\\
&&(0.003)&(0.004)&(0.003)&(0.002)&(<0.001)&(0.002)&(0.002)\\
\cline{2-9}
&\multirow{2}{*}{0.5}&0.878&0.89&\textbf{0.893}&\textbf{0.893}&0.527&\textbf{0.893}&0.876\\
&&(0.003)&(0.004)&(0.003)&(0.002)&(0.001)&(0.002)&(0.002)\\
\cline{2-9}
&\multirow{2}{*}{0.9}&0.89&0.91&0.92&\textbf{0.931}&0.535&0.903&0.884\\
&&(0.006)&(0.005)&(0.005)&(0.005)&(0.004)&(0.007)&(0.008)\\
\hline
\end{tabular}
}
\end{table}

\subsubsection{Evaluation under MNAR}

Real network results again suggest that our method is applicable for non-uniform sampling as long as assumption A1 holds. Results for the three networks are summarized in Tables \ref{AUC-nonuniform}--\ref{AUC-nonuniinv}. Similar to the synthetic experiments, when compared with the uniform-sampling setting, the positive setting, which puts more weight on high-degree nodes, improves the performance of all methods and the negative setting does the opposite.

\begin{table}[!ht]
\centering
\caption{\label{AUC-nonuniform}Predictive AUC on Three Networks, Positive MNAR Setting.}
\resizebox{0.65\textwidth}{!}{
\begin{tabular}{| c | c | c c c c c c c |}
\hline
Dataset&$\rho$&LE&SE&LE+&NS&ERGM&SBM&DCBM\\
\hline
\multirow{4}{*}{enron}&\multirow{2}{*}{0.2}&0.803&0.792&\textbf{0.822}&0.743&0.514&0.728&0.748\\
&&(0.003)&(0.003)&(0.003)&(0.003)&(0.003)&(0.003)&(0.003)\\
\cline{2-9}
&\multirow{2}{*}{0.5}&0.85&0.838&\textbf{0.864}&0.802&0.538&0.776&0.746\\
&&(0.004)&(0.005)&(0.004)&(0.003)&(0.004)&(0.005)&(0.005)\\
\hline
\multirow{4}{*}{UK faculty}&\multirow{2}{*}{0.2}&0.747&0.732&\textbf{0.774}&0.657&0.512&0.662&0.686\\
&&(0.004)&(0.006)&(0.004)&(0.005)&(0.002)&(0.008)&(0.006)\\
\cline{2-9}
&\multirow{2}{*}{0.5}&\textbf{0.802}&0.777&0.801&0.732&0.524&0.736&0.733\\
&&(0.008)&(0.008)&(0.007)&(0.006)&(0.004)&(0.006)&(0.008)\\
\hline
\multirow{4}{*}{US airports}&\multirow{2}{*}{0.2}&0.833&0.836&0.84&0.832&0.508&\textbf{0.866}&0.842\\
&&(0.003)&(0.005)&(0.004)&(0.002)&($<$0.001)&(0.002)&(0.002)\\
\cline{2-9}
&\multirow{2}{*}{0.5}&0.833&\textbf{0.849}&0.845&0.835&0.526&0.843&0.787\\
&&(0.003)&(0.003)&(0.003)&(0.003)&(0.001)&(0.004)&(0.006)\\
\hline
\end{tabular}
}
\end{table}

\begin{table}[!ht]
\centering
\caption{\label{AUC-nonuniinv}Predictive AUC on Three Networks, Negative MNAR Setting.}
\resizebox{0.65\textwidth}{!}{
\begin{tabular}{| c | c | c c c c c c c |}
\hline
Dataset&$\rho$&LE&SE&LE+&NS&ERGM&SBM&DCBM\\
\hline
\multirow{4}{*}{enron}&\multirow{2}{*}{0.2}&0.752&0.749&\textbf{0.781}&0.699&0.501&0.697&0.723\\
&&(0.004)&(0.003)&(0.003)&(0.003)&(0.003)&(0.003)&(0.003)\\
\cline{2-9}
&\multirow{2}{*}{0.5}&0.855&0.827&\textbf{0.867}&0.807&0.514&0.753&0.777\\
&&(0.003)&(0.002)&(0.002)&(0.003)&(0.003)&(0.003)&(0.002)\\
\hline
\multirow{4}{*}{UK faculty}&\multirow{2}{*}{0.2}&0.675&0.651&\textbf{0.709}&0.61&0.513&0.633&0.636\\
&&(0.008)&(0.009)&(0.007)&(0.007)&(0.002)&(0.008)&(0.006)\\
\cline{2-9}
&\multirow{2}{*}{0.5}&\textbf{0.837}&0.763&0.833&0.768&0.526&0.703&0.737\\
&&(0.004)&(0.005)&(0.004)&(0.005)&(0.003)&(0.007)&(0.006)\\
\hline
\multirow{4}{*}{US airports}&\multirow{2}{*}{0.2}&0.773&0.784&0.788&0.743&0.503&\textbf{0.851}&0.817\\
&&(0.006)&(0.007)&(0.006)&(0.006)&($<$0.001)&(0.002)&(0.003)\\
\cline{2-9}
&\multirow{2}{*}{0.5}&0.871&0.887&\textbf{0.892}&0.882&0.513&0.883&0.879\\
&&(0.002)&(0.003)&(0.003)&(0.002)&(0.001)&(0.002)&(0.002)\\
\hline
\end{tabular}
}
\end{table}

\section{DISCUSSION}

We have introduced an LE algorithm to fit low-rank models on egocentrically sampled partial networks. The approach is computationally efficient and presents theoretical guarantees for its correctness. Our technique is the first known consistent method for general low-rank models under egocentric sampling, and it does not require ``missing completely at random" assumptions. It can accurately predict missing links when the true model is low-rank or can be approximated by low-rank structures. However, we want to stress that one may have difficulty determining whether a partial network is from a low-rank model in practice; no single mode of link prediction works well in all cases. A practically preferable approach involves combining methods to achieve more adaptive performance in various situations, as studied by  \citet{peixoto2018reconstructing,ghasemian2020stacking,li2021network,yao2021bayesian}. Notably, even in this ensemble setting, having a strong individual method that works well in numerous situations is still necessary. We believe the proposed LE approach greatly contributes to potential candidates as one such model.

The proposed method can be extended in several directions. One limitation of our theory is the strict low-rank assumption; theoretical properties for approximately low-rank models are thus far unknown. A more general theory in these scenarios would largely expand the method’s scope. As another example, if a sequence of evolving networks is observed (subject to egocentric missingness), a critical but open question concerns how to fit a dynamic network model compatible with this missingness. We will leave this and other investigations for future work.

There are several applications in which we can potentially embed the current method. For example, graph learning methods such as GNN \citep{zhou2020graph} take networks as input, and our method can help to handle the semi-supervised prediction setting when the training data are only partially observed. Another application is to use our theory to study the privacy-preserving algorithms for network data \citep{hehir2021consistent}.

\subsubsection*{Acknowledgements}
T. Li is supported in part by the NSF grant DMS-2015298 and the University of Virginia 3-Caveliers award. G. Chan (as a student at the University of Virginia) is also supported in part by the NSF grant DMS-2015298.

\bibliographystyle{apalike}
\bibliography{bib}

\appendix
\section{Proofs}
In this section, we present the derivation of the error bound in the main paper. In our analysis, we will say some event happens with high probability if it happens with probability tending to 1 as $N\to \infty$. We use $C$ and $c$ as generic universal constants that may vary case by case. Let $\M{P},\M{A},\hat{\M{P}}$ be as defined as in the main paper, we introduce the following notations:
\begin{itemize}
    \item $\norm{\cdot}$ denotes the spectral norm
    \item $\norm{\cdot}_F$ denotes the Frobenius norm
    \item $p^* = \mathrm{max}_{ij}p_{ij}$
    \item $\lambda_k(\M{M})$ is the $k$-th largest eigenvalue of the matrix $M$
    \item $\sigma_k(\M{M})$ is the $k$-th largest singular value of the matrix $M$
\end{itemize}
\begin{lemma}[\citet{owen2009bi}]\label{A:lemma:owen}
For any $p\times q$ matrix $\M{M}$ with the partition
\begin{equation*}
    \M{M} = 
    \left(\begin{array}{cc}
    \mathbf{M}_{11}&\mathbf{M}_{12}\\
    \mathbf{M}_{21}&\mathbf{M}_{22}\\
    \end{array}\right),
\end{equation*}
Suppose $\rank(\M{M}_{11}) = \rank(\M{M})$, we have
    \begin{equation*}
        \mathbf{M}_{11}
        = \mathbf{M}_{12}\mathbf{M}_{22}^+\mathbf{M}_{21}.
    \end{equation*}
\end{lemma}
\begin{lemma}[\citet{lei2015consistency}]\label{lemma:concentration}
Let $\M{P}$ be the probability under the inhomogeneous Erd\"{o}s-Renyi model and $\M{A}$ be the adjacency matrix from $\M{P}$. Assume that $np^* \geq c\log{n}$ for some constant $c>0$. There exists a constant $C$ such that
    \begin{equation}
        \norm{\M{A}-\M{P}}\leq C\sqrt{np^*} 
    \end{equation}
    with high probability.
\end{lemma}
\begin{lemma}[\citet{yu2015useful}]\label{Davis}
Given a symmetric matrix $\M{P}$. Suppose $\rank(\M{P}) = K$ and let its eigendecomposition be ${\M{U}}{\M{\Sigma}}{\M{U}}^T$, where $\M{\Sigma} = \diag(\lambda_{1}, \cdots, \lambda_{K})$ contains all the eigenvalues in nonincreasing order. For another symmetric matrix $\M{A}$ in the same dimension, suppose its rank $K$ eigendecomposition is given by $\tilde{\M{U}}\tilde{\M{\Sigma}}\tilde{\M{U}}^T$.
There exists an orthogonal matrix $\M{O}\in\mathbb{R}^{K\times K}$ such that
\begin{equation}
    \norm{\tilde{\M{U}}\M{O}-\M{U}}_F 
    \leq\frac{3\sqrt{K}||A-P||}{\lambda_K}.
\end{equation}
\end{lemma}
\begin{lemma}[\citet{athreya2017statistical}]\label{lem:orthobound}
Let $\M{P}$ be the probability under the inhomogeneous Erd\"{o}s-Renyi model with $\rank(\M{P})= K$ and $\M{A}$ be the adjacency matrix from $\M{P}$. With the notations of Lemma~\ref{Davis}, and the same orthogonal matrix $\M{O}$, we have
\begin{equation}
    \norm{\M{O}\M{\Sigma}-\tilde{\M{\Sigma}}\M{O}}_F \le C(K^2+\log{n})
\end{equation}
with high probability.
\end{lemma}

\begin{ass}[Low-rank recoverable]\label{A:ass:rank}
The rank of the model satisifes $\rank(\M{P}_{11}) = \rank(\M{P}) = K$.
\end{ass}
\begin{ass}[Well-conditioned model]\label{A:ass:condition}
 There exists a constant $\psi>0$ such that
    $$\frac{1}{\psi}\cdot np^* \le \sigma_{K}(\M{P}_{11}) \le \sigma_{1}(\M{P}_{11}) \le \psi \cdot np^*$$
    $$\frac{1}{\psi}\cdot Np^* \le \sigma_{K}(\M{P}) \le \sigma_{1}(\M{P}) \le \psi \cdot Np^*$$
\end{ass}
Assumption~\ref{A:ass:rank} is strictly needed to ensure the validity of the low-rank recovery on the population matrix $\M{P}$ by Lemma~\ref{A:lemma:owen}. In contrast, assumptions \ref{A:ass:condition} can be relaxed for better generality. However, we keep it in the current form for conciseness and interpretability of our error bound. Moreover, both \ref{A:ass:rank} and \ref{A:ass:condition} are indeed motivated by the general sparse graphon model of \cite{bickel2009nonparametric} and are easy to hold when the egocentric sampling is done randomly on the nodes under many low-rank models. In particular, we have the following proposition 

\begin{theorem}\label{A:thm:frobenius}
Under assumptions~\ref{A:ass:rank} and \ref{A:ass:condition}, further assume that $np^*> c\log{n}$ and $K \le c\sqrt{\log n}$ for some constant $c>0$, we have  
\begin{equation}
    \norm{\hat{\M P}_{22}-\M P_{22}}_F \le C\sqrt{K}\left(\left(\frac{N}{n}\right)^{3/2}\sqrt{KNp^*} + \frac{N^2}{n^2}\log n\right)
\end{equation}
for some constant $C>0$ with high probability.
\end{theorem}

For illustration, consider the following two special cases
\begin{enumerate}
\item Suppose $n$ and $N$ are in the same order, we can see that the error bound on the missing network is in the order of $K\sqrt{np^*}+\sqrt{K}\log n$. Since $\norm{\M{P}_{22}}_F \sim n\sqrt{p^*}$, we know that $\norm{\hat{\M P}_{22}-\M P_{22}}_F/\norm{\M{P}_{22}}_F \to 0$ and the estimation consistency is guaranteed under the current assumptions.
\item Suppose $K$ is bounded and $np^* = \log^2 n$. Then the error bound is in the order of $\frac{N^2}{n^2}\log n$. So $\norm{\hat{\M P}_{22}-\M P_{22}}_F/\norm{\M{P}_{22}}_F \to 0$ as long as $n \gg N^{4/5}$. Therefore, though our method allows the sampling proportion to be vanishing, the decaying rate has to be slow.
\end{enumerate}

\begin{proof}[Proof of Theorem~\ref{A:thm:frobenius}]
Consider the prediction error of the probability matrix $\hat{\M{P}}_{22}-\M{P}_{22}$. We start with the spectral norm bound. By using the subproductivity of the spectral norm and triangular inequality, we have
\begin{align}\label{Pred-err-1}
    \norm{\hat{\M{P}}_{22}-\M{P}_{22}}=& \norm{\M{A}_{21}\tilde{\M{P}}_{11}^+\M{A}_{12}-\M{P}_{21}\M{P}_{11}^+\M{P}_{12}} \notag \\
    \leq& \norm{(\M{A}_{21}-\M{P}_{21})\M{P}_{11}^+\M{P}_{12}}+\norm{\M{A}_{21}(\tilde{\M{P}}_{11}^+\M{A}_{12}-\M{P}_{11}^+\M{P}_{12})} \notag \\
    \le & \norm{\M{A}_{21}-\M{P}_{21}} \norm{\M{P}_{11}^+}\norm{\M{P}_{12}}+ \norm{\M{A}_{21}}\norm{\tilde{\M{P}}_{11}^+\M{A}_{12}-\M{P}_{11}^+\M{P}_{12}} \notag\\
    \le &\norm{\M{A}_{21}-\M{P}_{21}} \norm{\M{P}_{11}^+}\norm{\M{P}_{12}}+ \norm{\M{A}_{21}}\norm{\tilde{\M{P}}_{11}^+\M{A}_{12}-\M{P}_{11}^+\M{A}_{12}} + \norm{\M{A}_{21}}\norm{\M{P}_{11}^+\M{A}_{12}-\M{P}_{11}^+\M{P}_{12}} \notag\\
    \le & \norm{\M{A}_{21}-\M{P}_{21}} \norm{\M{P}_{11}^+}\norm{\M{P}_{12}}+ \norm{\M{A}_{21}}\norm{\M{P}_{11}^+}\norm{\M{A}_{12}-\M{P}_{12}} + \norm{\M{A}_{21}}^2\norm{\tilde{\M{P}}_{11}^+-\M{P}_{11}^+}\notag \\
    \le & \norm{\M{A}_{12}-\M{P}_{12}}\norm{\M{P}_{11}^+}\left(2\norm{\M{P}_{12}}+\norm{\M{A}_{12}-\M{P}_{12}} \right) + \left(\norm{\M{P}_{21}}+ \norm{\M{A}_{12}-\M{P}_{12}} \right)^2\norm{\tilde{\M{P}}_{11}^+-\M{P}_{11}^+}\\
    = & \mathcal{I} + \mathcal{II}.
\end{align}
Denote the eigendecompositions up to $K$ of  $\M{P}_{11}$ and $\M{A}_{11}$ by $\M{P}_{11}=\M{U\Sigma U}^T$ and $\M{A}_{11}=\tilde{\M{U}}\tilde{\M{\Sigma}} \tilde{\M{U}}^T$ respectively. Note that since $\rank(\M{P}) = K$, the eigendecomposition of $\M{P}$ is exact. Note that, since the singular values match the eigenvalues up to their signs, we have $\M{P}_{11}^+=\M{U}\M{\Sigma}^{-1}\M{U}^T$ and $\tilde{\M{P}}_{11}^+=\tilde{\M{U}}\tilde{\M{\Sigma}}^{-1}\M{U}^T$.  We try to control the terms separately.

We want to control the concentration of each component of the $\M{A}$ matrix partition. In particular, we are taking the joint event of Lemma~\ref{lem:orthobound}, and Lemma~\ref{lemma:concentration} for $\M{A}$ and $\M{A}_{11}$. Notice that here $np^* > c \log{n}$ indicates that $Np^* > c\log{N}$ due to the monotonicity of $\log{n}/{n}$. Under this condition, therefore, we have $\norm{\M A_{11} - \M P_{11}} \le C\sqrt{np^*}$ and $\norm{\M A_{21} - \M P_{21}} \le \norm{\M A - \M P}  \le C\sqrt{Np^*}$. Under this event, we also have
$$|\lambda_{k}(\M{A}_{11}) - \lambda_{k}(\M{P}_{11})| \le \norm{\M{A}_{11} - \M{P}_{11}}\le \sqrt{np^*}, 1 \le k \le K.$$
Therefore, due to the assumption that $\lambda_{K}(\M{P}_{11}) \ge \psi np^*$ and $np^* \ge c \log n$, for sufficiently large $n$, we have
$$|\lambda_{K}(\M{A}_{11})| \ge \frac{1}{2}|\lambda_K(\M{P}_{11})|.$$

\paragraph{Upper bound of term $\mathcal{I}$.}

$$\norm{\M{P}_{11}^+} = \frac{1}{|\lambda_K(\M{P}_{11})|} \le \psi\cdot \frac{1}{np^*}.$$
Also notice that $\M{P}_{12}$ is a submatrix of $\M{P}$ so $\norm{\M P_{12}} \le \psi Np^*$. So we have
$$\mathcal{I} = \norm{\M{A}_{12}-\M{P}_{12}}\norm{\M{P}_{11}^+}\left(2\norm{\M{P}_{12}}+\norm{\M{A}_{12}-\M{P}_{12}} \right) \le C\frac{N}{n}\sqrt{Np^*}.$$

\paragraph{Upper bound of term $\mathcal{II}$.} Let $\M{O}\in\mathbb{R}^{K\times K}$ be an orthogonal matrix in Lemmas~\ref{Davis} and \ref{lem:orthobound}. Consider the term $\norm{\tilde{\M{P}}_{11}^+-\M{P}_{11}^+}$:
\begin{equation*}
    \begin{aligned}
    \norm{\tilde{\M P}_{11}^+-\M P_{11}^+}=&\norm{\tilde{\M U}\tilde{\M\Sigma}^{-1}\tilde{\M U}^T-\M U\M\Sigma^{-1}\M U^T}\\
    =&\norm{\tilde{\M U}\M O\M O^\top\tilde{\M \Sigma}^{-1}\tilde{\M U}^\top-\M U\M \Sigma^{-1}\M U^\top}\\
    \leq&\norm{\tilde{\M U}\M O-\M U}\norm{\M O^\top\tilde{\M \Sigma}^{-1}\tilde{\M U}^\top}+\norm{\M U\M O^\top\tilde{\M \Sigma}^{-1}\tilde{\M U}^\top-\M U\M \Sigma^{-1}\M U^\top}\\
    \leq&\norm{\tilde{\M U}\M O-\M U}\norm{\M O^\top\tilde{\M \Sigma}^{-1}\tilde{\M U}^\top}+\norm{ \M U\M O^\top\tilde{\M \Sigma}^{-1}\M O\M O^\top\tilde{\M U}^\top-\M U\M O^\top\tilde{\M \Sigma}^{-1}\M O\M U^\top}\\
    &+\norm{\M U\M O^\top\tilde{\M \Sigma}^{-1}\M O\M U^\top-\M U\M \Sigma^{-1}\M U^\top}\\
    \leq&\norm{\tilde{\M U}\M O-\M U}\norm{\tilde{\M \Sigma}^{-1}}+\norm{\M U\M O^\top\tilde{\M \Sigma}^{-1}\M O}\norm{\M O^\top\tilde{\M U}^\top-\M U^\top}+\norm{\M U}\norm{\M O^\top\tilde{\M \Sigma}^{-1}\M O - \M \Sigma^{-1}}\norm{\M U}\\
    \leq&2\norm{\tilde{\M U}\M O-\M U}\norm{\tilde{\M \Sigma}^{-1}}+\norm{\M O^\top\tilde{\M \Sigma}^{-1}\M O - \M \Sigma^{-1}}\\
    \leq&2\norm{\tilde{\M U}\M O-\M U}\norm{\tilde{\M \Sigma}^{-1}}+\norm{\tilde{\M \Sigma}^{-1}\M O-\M O\M \Sigma^{-1}}\\
    \leq&2\norm{\tilde{\M U}\M O-\M U}\norm{\tilde{\M \Sigma}^{-1}}+\norm{\tilde{\M \Sigma}^{-1}}\norm{\M \Sigma^{-1}}\norm{\M O\M \Sigma-\tilde{\M \Sigma}\M O}\\
    =&\norm{\tilde{\M \Sigma}^{-1}}\big(2\norm{\tilde{\M U}\M O-\M U}+\norm{\M \Sigma^{-1}}\norm{\M O\M \Sigma-\tilde{\M \Sigma}\M O}\big)\\
    \leq&[\frac{1}{2}\sigma_K(\M{P}_{11})]^{-1}\Big\{6\sqrt{\frac{K}{np^*}}+[\sigma_K(\M{P}_{11})]^{-1}C\mathrm{log}(n)\Big\}\\
    \le & C\frac{1}{np^*}(\sqrt{\frac{K}{np^*}} + \frac{\log n}{np^*}).
    \end{aligned}
\end{equation*}

Therefore, we have
$$\mathcal{II} = \left(\norm{\M{P}_{21}}+ \norm{\M{A}_{12}-\M{P}_{12}} \right)^2\norm{\tilde{\M{P}}_{11}^+-\M{P}_{11}^+} \le C\left(\left(\frac{N}{n}\right)^{3/2}\sqrt{KNp^*} + \frac{N^2}{n^2}\log n\right).$$
Note that this bound for $\mathcal{II}$ dominates that for $\mathcal{I}$. Substituting both bounds for $\mathcal{I}$ and $\mathcal{II}$ into \eqref{Pred-err-1} leads to
\begin{equation}
    \norm{\hat{\M P}_{22}-\M P_{22}}\leq C\left(\left(\frac{N}{n}\right)^{3/2}\sqrt{KNp^*} + \frac{N^2}{n^2}\log n\right).
\end{equation}
Finally, notice that $\rank(\hat{\M P}_{22})$ and $\rank(\M P_{22}) = K$, which indicates that $\rank(\hat{\M P}_{22}-\M P_{22}) \le 2K$. So we have the Frobenius norm bound
\begin{equation}
    \norm{\hat{\M P}_{22}-\M P_{22}}_F \leq C\sqrt{K}\left(\left(\frac{N}{n}\right)^{3/2}\sqrt{KNp^*} + \frac{N^2}{n^2}\log n\right).
\end{equation}
Finally, notice that if even the truncation to $[0,1]$ is applied, this process would not increase the error at all entries of $\hat{\M P}_{22}$, so the error bound still holds.
\end{proof}

The error bound for the full matrix recovery is a straightforward extension of Theorem~\ref{A:thm:frobenius}.
\begin{coro}\label{A:coro:full}
Under the assumptions of Theorem~\ref{A:thm:frobenius}, for the full matrix estimator $\hat{\M{P}}$, we have  
$$\norm{\hat{\M P}-\M P}_F \le C\sqrt{K}\left(\left(\frac{N}{n}\right)^{3/2}\sqrt{KNp^*} + \frac{N^2}{n^2}\log n\right)$$
with high probability.
\end{coro}

\begin{proof}[Proof of Corollary~\ref{A:coro:full}]
\begin{align*}
\norm{\hat{\M{P}} - \M P}_F^2 &= \norm{\hat{\M P}_{11}-\M P_{11}}_F^2+2 \norm{\hat{\M P}_{12}-\M P_{12}}_F^2 + \norm{\hat{\M P}_{22}-\M P_{22}}_F^2\\
&\le 2\norm{\hat{\M P}_{11}-\M P_{11}}_F^2+2 \norm{\hat{\M P}_{12}-\M P_{12}}_F^2 + \norm{\hat{\M P}_{22}-\M P_{22}}_F^2\\
 &\le 2\norm{\hat{\M P}_{\mathrm{obs}}-\M P_{\mathrm{obs}}}_F^2+ \norm{\hat{\M P}_{22}-\M P_{22}}_F^2\\
 &\le 2K\norm{\hat{\M P}_{\mathrm{obs}}-\M P_{\mathrm{obs}}}^2+ \norm{\hat{\M P}_{22}-\M P_{22}}_F^2.
\end{align*}

For the first term, under the same high probability event of Theorem~\ref{A:thm:frobenius}, we have 
\begin{align*}
\norm{\hat{\M P}_{\mathrm{obs}}-\M P_{\mathrm{obs}}} & \le \norm{\hat{\M P}_{\mathrm{obs}}-\M A_{\mathrm{obs}}} + \norm{{\M P}_{\mathrm{obs}}-\M A_{\mathrm{obs}}} \\
& \le \sigma_{K}(\M A_{\mathrm{obs}}) + \norm{{\M P}_{\mathrm{obs}}-\M A_{\mathrm{obs}}} \\
& \le \norm{{\M P}_{\mathrm{obs}}-\M A_{\mathrm{obs}}} + \sigma_{K}(\M P_{\mathrm{obs}}) + \norm{{\M P}_{\mathrm{obs}}-\M A_{\mathrm{obs}}}\\
& \le 2\norm{{\M P}-\M A} \\
& \le C\sqrt{Np^*}.
\end{align*}
Combining this result with Theorem~\ref{A:thm:frobenius}, we have
$$\norm{\hat{\M{P}} - \M P}_F^2 \le C\left( KNp^* + K^2\left(\frac{N}{n}\right)^{3}Np^* + K\left(\frac{N}{n}\right)^{4}\log^2{n}\right) \le C'\left( K^2\left(\frac{N}{n}\right)^{3}Np^* + K\left(\frac{N}{n}\right)^{4}\log^2{n}\right)$$
So we have
$$\norm{\hat{\M P}-\M P}_F \le C\sqrt{K}\left(\left(\frac{N}{n}\right)^{3/2}\sqrt{KNp^*} + \frac{N^2}{n^2}\log n\right)$$
under the event. From the proof, it can also be seen that the major error for the full matrix estimation is still on the unobserved component $\hat{\M P}_{22}$.

Finally, notice that if even the truncation to $[0,1]$ is applied, this process would not increase the error at all entries of $\hat{\M P}$, so the error bound still holds.

\end{proof}

\end{document}